\pdfoutput=1

\newif\ifdebug\debugtrue
\debugfalse
\newif\ifreview\reviewtrue
\reviewfalse

\documentclass[
  sigconf,
  authorversion,
  natbib=false,
  anonymous=\ifreview true\else false\fi,
  review=\ifreview true\else false\fi,
  authordraft=\ifdebug true\else false\fi,
  timestamp=\ifdebug true\else false\fi,
]{acmart}
\usepackage[utf8]{inputenc}

\RequirePackage[
  abbreviate=true,
  dateabbrev=true,
  isbn=true,
  doi=true,
  urldate=comp,
  url=true,
  maxnames=2,
  backref=false,
  backend=biber,
  style=ACM-Reference-Format,
  language=american
]{biblatex}
\usepackage[leqno]{amsmath}
\usepackage{amsmath,amssymb,amsthm}
\theoremstyle{remark}
\newtheorem*{remark}{Remark}
\usepackage{booktabs}  
\usepackage{algorithm}
\usepackage[noend]{algpseudocode}
\usepackage{array}
\usepackage{xfrac}
\usepackage{microtype}
\usepackage{booktabs}  
\usepackage{xparse}
\usepackage{xspace}

\usepackage{bm}
\usepackage{hyperref}
\usepackage{gensymb}

\newcommand{\op}[1]{\ensuremath{\operatorname{#1}}}

\def\abbr#1{\textsc{\MakeLowercase{#1}}}
\let\term\emph

\makeatletter
\newcommand*{\bigominus}{\DOTSB\bigominus@\slimits@}
\newcommand{\bigominus@}{\mathop{\mathpalette\bigominus@@\relax}}
\newcommand{\bigominus@@}[2]{%
  \vcenter{\hbox{%
    \sbox\z@{$\m@th#1\bigoplus$}%
    \resizebox{\wd\z@}{!}{$\m@th#1\ominus$}%
  }}%
}
\makeatother

\NewDocumentCommand\todo{ggo}{%
  \ifdebug
    \textcolor{red}{TODO}%
    \IfValueT{#1}{\textcolor{red}{:~}}%
    \IfValueT{#1}{\textcolor{blue}{#1}}\nopagebreak%
    \IfValueT{#2}{\\\textcolor{orange}{#2}}\nopagebreak%
    \IfValueT{#3}{\\\textcolor{magenta}{#3}}%
  \fi}
\NewDocumentCommand\done{ggo}{%
  \ifdebug
    \textcolor{green}{DONE}%
    \IfValueT{#1}{\textcolor{green}{:~}}%
    \IfValueT{#1}{\textcolor{blue}{\sout{#1}}}\nopagebreak%
    \IfValueT{#2}{\\\textcolor{orange}{\sout{#2}}}\nopagebreak%
    \IfValueT{#3}{\\\textcolor{magenta}{\sout{#3}}}%
  \fi}
\newcommand*{\tran}{^{\mkern-1.5mu\mathsf{T}}}
\let\emph=\textit
\let\newcite=\textcite

\algnewcommand\algorithmicforeach{\textbf{for each}}
\algdef{S}[FOR]{ForEach}[1]{\algorithmicforeach\ #1\ \algorithmicdo}

\setcopyright{acmcopyright}

\title{Implementation Notes for the Soft Cosine Measure}

\author{V\'it Novotn\'y}
\orcid{0000-0002-3303-4130}
\affiliation{%
  \institution{Masaryk University}
  \department{Faculty of Informatics}
  \streetaddress{Botanick\'a 68a}
  \city{Brno}
  \state{Czech Republic}
  \postcode{602 00}
}
\email{witiko@mail.muni.cz}

\acmConference[CIKM '18]{%
  The 27th ACM International Conference on Information and Knowledge
  Management}{October 22--26, 2018}{Torino, Italy}
\acmBooktitle{%
  The 27th ACM International Conference on Information and Knowledge Management
  (CIKM '18), October 22--26, 2018, Torino, Italy}
\acmDOI{10.1145/3269206.3269317}
\acmISBN{978-1-4503-6014-2/18/10}
\acmPrice{15.00}
\acmYear{2018}
\copyrightyear{2018}

\keywords{%
  Vector Space Model,
  computational complexity,
  similarity measure
}

\ifdebug
\usepackage[normalem]{ulem}
\usepackage[hybrid,citations]{markdown}

\fi
\begin{document}

\begin{abstract}
The standard bag-of-words vector space model (\abbr{VSM}) is efficient, and
ubiquitous in information retrieval, but it underestimates the
similarity of documents with the same meaning, but different terminology. To
overcome this limitation, \newcite{sidorov2014soft} proposed the Soft Cosine
Measure (\abbr{SCM}) that incorporates term similarity relations.
\newcite{charletdamnati17} showed that the \abbr{SCM} is highly effective in
question answering (\abbr{QA}) systems. However, the orthonormalization
algorithm proposed by \newcite{sidorov2014soft} has an impractical time
complexity of $\mathcal O(n^4)$, where $n$ is the size of the vocabulary.
\looseness=-1

In this paper, we prove a tighter lower worst-case time complexity bound of
$\mathcal O(n^3)$. We also present an algorithm for computing the similarity
between documents and we show that its worst-case time complexity is $\mathcal
O(1)$ given realistic conditions. Lastly, we describe implementation in
general-purpose vector databases such~as Annoy, and Faiss and in the inverted
indices of text search engines such~as Apache Lucene, and ElasticSearch.
Our results enable the deployment of the \abbr{SCM} in real-world information
retrieval systems.
\end{abstract}

\maketitle

\section{Introduction}
\label{sec:similarity-introduction}
The standard bag-of-words vector space model
(\abbr{VSM})~\cite{ml:SaltonBuckley1988}\index{standard model}
represents documents as real vectors.
Documents are expressed in a basis where each basis vector corresponds to a
single term, and each coordinate corresponds to the frequency of a term in a
document. Consider the documents
\begin{align*}
  d_1 &=\text{“When Antony found \textbf{Julius Caesar} dead”, and} \\
  d_2 &=\text{“I did enact \textbf{Julius Caesar}: I was killed i' the Capitol”}
\end{align*}
represented in a basis $\{\bm\alpha_i\}_{i=1}^{14}$\index{.a@$\bm\alpha$|emph} of
$\mathbb R^{14}$, where the basis vectors correspond
to the terms in the order of first appearance. Then the corresponding document vectors
$\mathbf v_1$, and $\mathbf v_2$ would have the following coordinates in $\bm\alpha$:
{\linepenalty=500
\begin{align*}
  (\mathbf v_1)_{\bm\alpha} &= [1\:1\:1\:\textbf1\:\textbf1\:1\:0\:0\:0\:0\:0\:0\:0\:0]\tran, \text{and} \\
  (\mathbf v_2)_{\bm\alpha} &= [0\:0\:0\:\textbf1\:\textbf1\:0\:2\:1\:1\:1\:1\:1\:1\:1]\tran.
\end{align*}}%
Assuming $\bm\alpha$ is orthonormal, we can take the inner product of the
$\ell^2$-normalized vectors $\mathbf v_1$, and $\mathbf v_2$ to measure the
cosine of the angle
(i.e.\ the \term{cosine similarity}) between the documents $d_1$, and~$d_2$:
\begin{equation*}
  \langle\mathbf v_1/\Vert \mathbf v_1\Vert, \mathbf v_2/\Vert\mathbf
  v_2\Vert\rangle = \frac{\bigl((\mathbf v_1)_{\bm\alpha}\bigr)\tran (\mathbf v_2)_{\bm\alpha}}{\sqrt{\bigl((\mathbf v_1)_{\bm\alpha}\bigr)\tran (\mathbf
  v_1)_{\bm\alpha}}\sqrt{\bigl((\mathbf v_2)_{\bm\alpha}\bigr)\tran (\mathbf
  v_2)_{\bm\alpha}}}\approx0.23.
\end{equation*}
Intuitively, this underestimates the true similarity between $d_1$, and~$d_2$.
Assuming $\bm\alpha$ is orthogonal but
not orthonormal, and that the terms Julius, and Caesar are twice as
important as the other terms, we can construct a diagonal change-of-basis
matrix $\mathbf W = (w_{ij})$ from $\bm\alpha$ to an
orthonormal basis $\bm\beta$, where $w_{ii}$ corresponds
to the importance of a term~$i$. This brings us closer to the true similarity:
\begin{eqnarray*}
  (\mathbf v_1)_{\bm\beta} &\!=\!& \mathbf W (\mathbf v_1)_{\bm\alpha} = [1\:1\:1\:\textbf2\:\textbf2\:1\:0\:0\:0\:0\:0\:0\:0\:0]\tran, \\
  (\mathbf v_2)_{\bm\beta} &\!=\!& \mathbf W (\mathbf v_2)_{\bm\alpha} = [0\:0\:0\:\textbf2\:\textbf2\:0\:2\:1\:1\:1\:1\:1\:1\:1]\tran, \text{and}
\end{eqnarray*}
\begin{multline*}
 \langle\mathbf v_1/\Vert \mathbf v_1\Vert, \mathbf v_2/\Vert\mathbf v_2\Vert\rangle\\
   = 
  \frac{\bigl(\mathbf W(\mathbf v_1)_{\bm\alpha}\bigr)\tran \mathbf W(\mathbf v_2)_{\bm\alpha}}
  {\sqrt{\bigl(\mathbf W(\mathbf v_1)_{\bm\alpha}\bigr)\tran \mathbf W(\mathbf
    v_1)_{\bm\alpha}}\sqrt{\bigl(\mathbf W(\mathbf v_2)_{\bm\alpha}\bigr)\tran \mathbf W(\mathbf
    v_2)_{\bm\alpha}}}\approx0.53.
\end{multline*}
Since we assume that the bases $\bm\alpha$ and $\bm\beta$ are orthogonal, the
terms dead and
killed contribute nothing to the cosine similarity despite the clear synonymy,
because $\langle \bm\beta_{\text{dead}}, \bm\beta_{\text{killed}}\rangle=0$.
In general, the \abbr{VSM} will underestimate the true similarity between
documents that carry the same meaning but use different terminology.

In this paper, we further develop the soft \abbr{VSM} described by
\newcite{sidorov2014soft}, which does not assume $\bm\alpha$ is orthogonal
and which achieved state-of-the-art results on the question
answering~(\abbr{QA}) task at SemEval~2017~\cite{charletdamnati17}.  In
Section~\ref{sec:similarity-relwork}, we review the previous work incorporating
term similarity into the \abbr{VSM}. In
Section~\ref{sec:similarity-generalized-vsm}, we restate the
definition of the soft \abbr{VSM} and present several computational
complexity results. In Section~\ref{sec:similarity-implementation}, we describe
the implementation in vector databases and
inverted indices. We conclude in Section~\ref{sec:similarity-conclusion} by
summarizing our results and suggesting future work.\looseness=-1

\section{Related work}
\label{sec:similarity-relwork}
Most works incorporating term similarity into the \abbr{VSM} published prior to
\newcite{sidorov2014soft} remain in an
orthogonal coordinate system and instead propose novel document similarity
measures. To name a few, \newcite{mikawa2011proposal} proposes the
\term{extended cosine measure}, which introduces a metric matrix~$\mathbf Q$ as a
multiplicative factor in the cosine similarity formula. $\mathbf Q$ is the
solution of an optimization problem to maximize the sum of extended cosine
measures between each vector and the centroid of the vector's category.
Conveniently, the metric matrix~$\mathbf Q$ can be used directly with the soft
\abbr{VSM}, where
it defines the inner product between basis vectors. \newcite{jimenez2012soft}
equip the multiset \abbr{VSM} with a \term{soft cardinality} operator that
corresponds to cardinality, but takes term similarities into account.

The notion of generalizing the \abbr{VSM} to non-orthogonal coordinate systems was
perhaps first explored by \newcite{sidorov2014soft} in the context of entrance
exam question answering, where the basis vectors did not correspond directly to
terms, but to $n$-grams constructed by following paths in syntactic
%
trees. The authors derive the inner product of two basis vectors from the edit
distance\index{edit distance} between the corresponding $n$-grams.
\term{Soft cosine measure} (\abbr{SCM}) is how they term the formula for computing the
cosine similarity between two vectors expressed in a non-orthogonal basis.
They also present an algorithm that computes a change-of-basis matrix to an
orthonormal basis in time $\mathcal{O}(n^4)$. We present an $\mathcal{O}(n^3)$
algorithm in this~paper.

\newcite{charletdamnati17} achieved state-of-the-art results at the \abbr{QA}
task at SemEval~2017~\cite{nakov2017semeval} by training a
document classifier on soft cosine measures between document passages. Unlike
\newcite{sidorov2014soft}, \newcite{charletdamnati17} already use basis vectors
that correspond to terms rather than to $n$-grams. They derive the inner
product of two basis vectors both from the edit distance between the
corresponding terms, and from the inner product of the corresponding
word2vec term embeddings~\cite{mikolov2013efficient}.

\section{Computational complexity}
\label{sec:similarity-generalized-vsm}
In this section, we restate the definition of the soft \abbr{VSM}
as it was described by \newcite{sidorov2014soft}. We then prove a tighter lower
worst-case time complexity bound for computing a change-of-basis matrix to an
orthonormal basis. We also prove that under certain assumptions, the inner
product is a linear-time operation.

\begin{definition}
Let $\mathbb R^n$ be the real $n$-space over $\mathbb R$
equipped with the bilinear inner product $\langle\cdot,\cdot\rangle$. Let
$\{\bm{\alpha}_i\}_{i = 1}^n$ be the basis
of $\mathbb R^n$ in which we express our vectors. Let $\mathbf
W_{\bm{\alpha}}=(w_{ij})$ be a diagonal change-of-basis matrix from
$\bm\alpha$ to a normalized basis $\{\bm{\beta}_i\}_{i =
1}^n$ of $\mathbb R^n$, i.e.\ $\langle
\bm{\beta}_i, \bm{\beta}_j\rangle\in[-1, 1], \langle \bm{\beta}_i,
\bm{\beta}_i\rangle=1$. Let $\mathbf S_{\bm{\beta}}=(s_{ij})$\index{.s@$\mathbf
S$|emph} be the metric matrix of $\mathbb R^n$ w.r.t.\ $\bm{\beta}$,
i.e.\ $s_{ij} = \langle \bm{\beta}_i, \bm{\beta}_j\rangle$. Then $(\mathbb
R^n, \mathbf W_{\bm{\alpha}}, \mathbf S_{\mathbf\beta})$ is a
\textit{soft~\abbr{VSM}}.\looseness=-1
\end{definition}

\begin{theorem}
\label{thm:similarity-orthonormalization}
Let\,\,$G=(\mathbb R^n, \mathbf W_{\bm{\alpha}}, \mathbf S_{\bm{\beta}})$ be a soft
\abbr{VSM}. Then a change-of-basis matrix $\mathbf E$ from the basis $\bm\beta$ to an orthonormal
basis of\ \,$\mathbb R^n$ can be computed in time $\mathcal O(n^3)$.
\end{theorem}

\begin{proof}
By definition, $\mathbf S=\mathbf E\mathbf E\tran$ for any change-of-basis
matrix $\mathbf E$ from the basis $\bm\beta$ to an orthonormal basis. Since
$\mathbf S$ contains inner products of linearly independent vectors
$\bm{\beta}$, it is Gramian and positive definite
\cite[p.~441]{horn2013matrix}. The Gramianness of $\mathbf S$ also implies
its symmetry. Therefore, a lower triangular~$\mathbf E$ is uniquely
determined by the Cholesky factorization of the symmetric positive-definite~$\mathbf S$, which we can compute in time
$\mathcal{O}(n^3)$~\cite[p.~191]{stewart1998matrix}.
\end{proof}

\begin{remark}
See Table~\ref{tab:speed-evaluation} for an experimental comparison.\looseness=-1

Although the vocabulary in our introductory example contains only $n=14$ terms,
$n$ is in the millions for real-world corpora such as the English Wikipedia.
Therefore, we generally need to store the $n\times n$ matrix $\mathbf S$
in a sparse format, so that it fits into main memory. Later, we will discuss
how the density of $\mathbf S$ can be reduced, but the Cholesky factor
$\mathbf E$ can also be arbitrarily dense and therefore expensive to store.
Given a permutation matrix $\mathbf P$, we can instead
factorize $\mathbf{P\tran SP}$ into $\mathbf{FF\tran}$. Finding the permutation
matrix~$\mathbf P$ that minimizes the density of the
Cholesky factor $\mathbf F$ is \textsf{NP}-hard
\cite{yannakakis1981computing}, but heuristic stategies are known
\cite{cuthill1969reducing,heggernes2001computational}. Using the fact
that $\mathbf P\tran=\mathbf P^{-1}$, and basic facts about transpose,
we can derive $\mathbf{E}=\mathbf{PF}$ as follows:
$
  \mathbf S = \mathbf{PP\tran SPP\tran} = \mathbf{PFF\tran P\tran} = \mathbf{PF(PF)\tran} = \mathbf{EE\tran}.
$
\end{remark}

\begin{table}[tb]
\centering
\caption{\textmd{The real time to compute a matrix $\mathbf E$ from a
  dense matrix~$\mathbf{S}$ averaged over 100 iterations. We used two Intel
  Xeon E5-2650 v2 (20M cache, 2.60~GHz) processors to evaluate the $\mathcal
  O(n^3)$ Cholesky factorization from NumPy 1.14.3, and the $\mathcal O(n^4)$
  iterated Gaussian elimination from \abbr{LAPACK}. For $n>1000$, only sparse
  $\mathbf{S}$ seem practical.}\looseness=-1}
\label{tab:speed-evaluation}
\begin{tabular}{rlr@{.}lr@{.}l}
\textbf{$n$ terms} & \textbf{Algorithm} & \multicolumn{2}{l}{\textbf{Real computation time}} \\ \midrule
 100 & Cholesky factorization &   0&0006\,sec (0.606\,ms)   \\
 100 & Gaussian elimination   &   0&0529\,sec (52.893\,ms)  \\
 500 & Cholesky factorization &   0&0086\,sec (8.640\,ms)   \\
 500 & Gaussian elimination   &  22&7361\,sec (22.736\,sec) \\
1000 & Cholesky factorization &   0&0304\,sec (30.378\,ms)  \\
1000 & Gaussian elimination   & 354&2746\,sec (5.905\,min)  \\
\end{tabular}
\end{table}

\begin{lemma}
\label{lemma:similarity-inner-product}
Let $G=(\mathbb R^n, \mathbf W_{\bm{\alpha}}, \mathbf S_{\bm{\beta}})$ be a
soft \abbr{VSM}. Let $\mathbf x,\mathbf y\in\mathbb R^n$. Then $\langle
\mathbf x,\mathbf y\rangle=(\mathbf W(\mathbf x)_{\bm{\alpha}})\tran\mathbf S\mathbf
W(\mathbf y)_{\bm{\alpha}}$.
\end{lemma}

\begin{proof}
Let $\mathbf E$ be the change-of-basis matrix from the basis $\bm{\beta}$ to an
orthonormal basis $\bm{\gamma}$ of $\mathbb R^n$. Then:
\begin{multline*}
  \label{eq:similarity-inner-product}
  \langle\mathbf x, \mathbf y\rangle
  = \bigl((\mathbf{x})_{\bm\gamma}\bigr)\tran (\mathbf{y})_{\bm\gamma}
  = \bigl(\mathbf{E(x)}_{\bm\beta}\bigr)\tran \mathbf{E(y)}_{\bm\beta}
  = \bigl(\mathbf{EW(x)}_{\bm\alpha}\bigr)\tran \mathbf{EW(y)}_{\bm\alpha} \\
  = \bigg(\sum_{i=1}^n (\bm{\alpha}_i)_{\bm\gamma}\cdot w_{ii}\cdot(x_i)_{\bm\alpha}\bigg)\cdot\bigg(\sum_{j=1}^n(\bm{\alpha}_j)_{\bm\gamma}\cdot w_{jj}\cdot (y_j)_{\bm\alpha}\bigg)
  \\[-0.5ex]
  = \sum_{i=1}^n\sum_{j=1}^nw_{ii}\cdot(x_i)_{\bm\alpha}\cdot\langle\bm\alpha_i, \bm\alpha_j\rangle\cdot w_{jj}\cdot (y_j)_{\bm\alpha} \\[-0.5ex]
  = \sum_{i=1}^n\sum_{j=1}^nw_{ii}\cdot(x_i)_{\bm\alpha}\cdot s_{ij}\cdot w_{jj}\cdot (y_j)_{\bm\alpha}
  = \bigl(\mathbf{W(x)}_{\bm\alpha}\bigr)\tran \mathbf S\mathbf{W(y)}_{\bm\alpha}.\!\!\qedhere
\end{multline*}
\end{proof}

\begin{remark}
From here, we can directly derive the cosine of the angle between $\mathbf x$
and $\mathbf y$ (i.e.\ what \newcite{sidorov2014soft} call the \abbr{SCM}) as
follows:\looseness=-1
\begin{equation*}
  \langle\mathbf x/\Vert\mathbf x\Vert, \mathbf y/\Vert\mathbf y\Vert\rangle
  = \frac{\bigl(\mathbf{W(x)}_{\bm\alpha}\bigr)\tran\mathbf S\mathbf{W(y)}_{\bm\alpha}}{\sqrt{\bigl(\mathbf{W(x)}_{\bm\alpha}\bigr)\tran\mathbf S\mathbf{W(x)}_{\bm\alpha}}
 \sqrt{\bigl(\mathbf{W(y)}_{\bm\alpha}\bigr)\tran\mathbf S\mathbf{W(y)}_{\bm\alpha}}}.
\end{equation*}
The \abbr{SCM} is actually the starting point for
\newcite{charletdamnati17}, who propose matrices $\mathbf S$ that are not
necessarily metric. If, like them, we are only interested in computing the
\abbr{SCM}, then we only require that the square roots remain real, i.e.\
that $\mathbf x\not=0\implies(\mathbf{W(x)}_{\bm\alpha})\tran\mathbf
S\mathbf{W(x)}_{\bm\alpha}\geq 0$.  For arbitrary $\mathbf x\in\mathbb{R}^n$,
this holds iff $\mathbf S$ is positive semi-definite. However, since the
coordinates $(\mathbf x)_{\bm\alpha}$ correspond to non-negative term frequencies,
it is sufficient that $\mathbf W$ and $\mathbf S$ are
non-negative as well. If we are only interested in computing the inner product,
then $\mathbf S$ can~be~arbitrary.
\end{remark}

\begin{theorem}
\label{thm:similarity-inner-product}
Let $G=(\mathbb R^n, \mathbf W_{\bm{\alpha}}, \mathbf S_{\bm{\beta}})$ be a
soft \abbr{VSM} such that no column of\ \,$\mathbf S$ contains more than $C$
non-zero elements, where $C$ is a constant. Let $\mathbf x,\mathbf y\in\mathbb
R^n$ and let $m$ be the number of non-zero elements in $(\mathbf x)_{\bm\beta}$.
Then $\langle\mathbf x, \mathbf y\rangle$ can be computed in time $\mathcal O(m)$.
\end{theorem}

\begin{proof}
Assume that $(\mathbf x)_{\bm{\alpha}}, (\mathbf y)_{\bm{\alpha}},$
and $\mathbf S$ are represented by data structures with constant-time
column access and non-zero element traversal, e.g.\ compressed sparse column
(\abbr{CSC}) matrices.
Further assume that $\mathbf W$ is represented by an array containing the main
diagonal of $\mathbf W$. Then
Algorithm~\ref{algo:similarity-inner-product} computes $\bigl(\mathbf W(\mathbf
x)_{\bm{\alpha}}\bigr)\tran\mathbf S\mathbf W(\mathbf y)_{\bm{\alpha}}$ in time
$\mathcal{O}(m)$, which by Lemma~\ref{lemma:similarity-inner-product},
corresponds to $\langle\mathbf x, \mathbf y\rangle$.\qedhere

\begin{algorithm}
\caption{The inner product of $\mathbf x$ and $\mathbf y$}
\label{algo:similarity-inner-product}
\begin{algorithmic}[1]
\State $r\gets 0$
\ForEach {$i$ such that $(x_i)_{\bm\alpha}$ is non-zero}
\Comment{$= m$ iterations}
\ForEach {$j$ such that $s_{ij}$ is non-zero}
\Comment{$\leq C$ iterations}
\State $r\gets r+w_{ii}\cdot(x_i)_{\bm\alpha}\cdot s_{ij}\cdot w_{jj}\cdot (y_j)_{\bm\alpha}$
\EndFor
\EndFor
\vspace*{-1.09ex}
\State \textbf{return} $r$
\end{algorithmic}
\end{algorithm}
\vspace*{-2ex}
\end{proof}
\vspace*{-2ex}

\begin{remark}
Similarly, we can show that if a column of $\mathbf S$ contains $C$ non-zero
elements on average, $\langle\mathbf x, \mathbf y\rangle$ has the average-case
time complexity of $\mathcal O(m)$. Note also that most information
retrieval systems impose a limit on the length of a query document. Therefore,
$m$ is usually bounded by a constant and $\mathcal O(m)=\mathcal O(1)$.

Since we are usually interested in the inner products of all document pairs
in two corpora (e.g. one containing queries and the other actual documents),
we can achieve significant speed improvements with vector processors by
computing $(\mathbf{WX})\tran\mathbf{SWY}$, where $\mathbf{X}$, and
$\mathbf{Y}$ are \emph{corpus matrices} containing the coordinates of document
vectors in the basis $\bm{\alpha}$ as columns. To compute the \abbr{SCM},
we first need to normalize the document vectors by performing an
entrywise division of every column in $\mathbf{X}$ by
$\op{diag}\sqrt{(\mathbf{WX})\tran\mathbf{SWX}}=\sqrt{(\mathbf{WX})\tran\mathbf{S}\circ(\mathbf{WX})\tran},$
where $\circ$ denotes entrywise product. $\mathbf{Y}$ is normalized
analogously.

There are several strategies for making no column of $\mathbf S$ contain more
than $C$ non-zero elements.
If we do not require that $\mathbf S$ is metric (e.g.\ because we only wish to
compute the inner product, or the \abbr{SCM}), a simple strategy is to start with an empty
matrix, and to insert the $C-1$ largest elements and the diagonal element from
every column of $\mathbf S$. However, the resulting matrix will likely be
asymmetric, which makes the inner product formula asymmetric as well.
We can regain symmetry by always inserting an element $s_{ij}$
together with the element $s_{ji}$ and only if this does not make the column
$j$ contain more than $C$ non-zero elements. This strategy is greedy, since
later columns contain non-zero elements inserted by earlier columns. Our
preliminary experiments suggest that processing colums that correspond to
increasingly frequent terms performs best on the task of
\newcite{charletdamnati17}.
Finally, by limiting the sum of all non-diagonal elements in a column to be
less than one, we can make $\mathbf S$ strictly diagonally dominant and
therefore positive definite, which enables us to compute $\mathbf E$ through
Cholesky factorization.
\end{remark}

\section{Implementation in vector databases and inverted indices}
\label{sec:similarity-implementation}
\looseness=-1
In this section, we present coordinate transformations for retrieving nearest
document vectors according to the inner product, and the soft cosine measure
from general-purpose vector databases such as Annoy, or Faiss~\cite{JDH17}. We
also discuss the implementation in the inverted indices of text search engines
such as Apache Lucene~\cite{bialecki12}.

\begin{remark}
With a vector database, we can transform document vectors to an orthonormal
basis $\bm{\gamma}$. In the transformed coordinates, the dot product
$((\mathbf{x})_{\bm\gamma})\tran (\mathbf{y})_{\bm\gamma}$ corresponds to the
inner product $\langle\mathbf{x}, \mathbf{y}\rangle$ and the cosine similarity
corresponds to the cosine of an angle $\langle\mathbf x/\Vert\mathbf x\Vert,
\mathbf y/\Vert\mathbf y\Vert\rangle$ (i.e. the soft cosine measure). A vector
database that supports nearest neighbor search according to either the dot
product, or the cosine similarity will therefore retrieve vectors expressed in
$\bm{\gamma}$ according to either the inner product, or the soft cosine
measure. 
We can compute a change-of-basis matrix~$\mathbf E$ of order $n$ in time
$\mathcal{O}(n^3)$ by Theorem~\ref{thm:similarity-orthonormalization} and use it to
transform every vector $\mathbf x\in\mathbb R^n$ to $\bm\gamma$ by computing
$\mathbf{EW(x)_{\bm\alpha}}$. However, this approach requires that $\mathbf S$
is symmetric positive-definite and that we recompute $\mathbf E$, and reindex the
vector database each time $\mathbf S$ has changed. We will now discuss
transformations that do not require $\mathbf E$ and for which a non-negative
$\mathbf S$ is sufficient as discussed in the remark for
Lemma~\ref{lemma:similarity-inner-product}.\looseness=-1
\end{remark}

\begin{theorem}
Let $G=(\mathbb R^n, \mathbf W_{\bm{\alpha}}, \mathbf S_{\bm{\beta}})$ be a
soft \abbr{VSM}. Let $\mathbf{x,x',y}\in\mathbb R^n$ such
that $(\mathbf{x}')_{\bm{\beta}} =
\mathbf{S}\tran(\mathbf{x})_{\bm{\beta}}.$ Then $\langle \mathbf x, \mathbf
y\rangle=((\mathbf{x}')_{\bm{\beta}})\tran(\mathbf{y})_{\bm{\beta}}.$
\end{theorem}
{\linepenalty=500\par}

\begin{proof}
$\bigl((\mathbf{x}')_{\bm{\beta}}\bigr)\tran(\mathbf{y})_{\bm{\beta}}\!=\!\bigl((\mathbf x)_{\bm\beta}\bigr)\tran\mathbf S(\mathbf y)_{\bm\beta}\!=\!\langle\mathbf x,
\mathbf y\rangle$ from Lemma\,\ref{lemma:similarity-inner-product}\rlap.\!\!\!\looseness=-1
\end{proof}

\begin{remark}
By transforming a query vector $\mathbf x$ into $(\mathbf x')_{\bm\beta}$,
we can retrieve documents according to the inner product in vector databases
that only support nearest neighbor search according to the dot product.
Note that we do not introduce $\mathbf S$ into $(\mathbf y)_{\bm\beta}$, which
allows us to change $\mathbf S$ without changing the documents in a vector
database and that $\mathbf S$ can be arbitrary as discussed in the
remark for Lemma~\ref{lemma:similarity-inner-product}.
\end{remark}

\begin{theorem}
\label{theorem:similarity-partial}
Let $G=(\mathbb R^n, \mathbf W_{\bm{\alpha}}, \mathbf S_{\bm{\beta}})$
be a soft \abbr{VSM}.
Let $\mathbf{x,x',y,y',}$\penalty-500$\mathbf{z,z'}\in\mathbb R^n$ s.t.\ $
\mathbf{x},\mathbf{y},\mathbf{z}\not=0,
(\mathbf{x}')_{\bm{\beta}} = \mathbf{S}\tran(\mathbf{x})_{\bm{\beta}},
(\mathbf{y}')_{\bm{\beta}} = \frac{(\mathbf{y})_{\bm{\beta}}}{\sqrt{\bigl((\mathbf{y})_{\bm{\beta}}\bigr)\tran\mathbf{S}(\mathbf{y})_{\bm{\beta}}}},$
and\ \,$
(\mathbf{z}')_{\bm{\beta}} = \frac{(\mathbf{z})_{\bm{\beta}}}{\sqrt{\bigl((\mathbf{z})_{\bm{\beta}}\bigr)\tran\mathbf{S}(\mathbf{z})_{\bm{\beta}}}}.$
Then $
\langle\mathbf x/\Vert\mathbf x\Vert, \mathbf y/\Vert\mathbf y\Vert\rangle\leq\langle\mathbf x/\Vert\mathbf x\Vert,\mathbf z/\Vert\mathbf z\Vert\rangle
$ iff\ \,$
\bigl((\mathbf{x}')_{\bm{\beta}}\bigr)\tran(\mathbf{y}')_{\bm{\beta}}\leq
\bigl((\mathbf{x}')_{\bm{\beta}}\bigr)\tran(\mathbf{z}')_{\bm{\beta}}$.
\end{theorem}
{\linepenalty=500\par}

\begin{proof}
$\bigl((\mathbf{x}')_{\bm{\beta}}\bigr)\tran(\mathbf{y'})_{\bm{\beta}} = \frac{((\mathbf x)_{\bm\beta})\tran\mathbf S(\mathbf y)_{\bm\beta}}{\sqrt{((\mathbf y)_{\bm\beta})\tran\mathbf S(\mathbf y)_{\bm\beta}}}$. From
Lemma~\ref{lemma:similarity-inner-product}, this equals $\langle\mathbf
x/\Vert\mathbf x\Vert, \mathbf y/\Vert\mathbf y\Vert\rangle$ except for the missing
term $\sqrt{\bigl((\mathbf x)_{\bm\beta}\bigr)\tran\mathbf S(\mathbf x)_{\bm\beta}}$ in the
divisor. The term is constant in both $\langle\mathbf x/\Vert\mathbf
x\Vert,\mathbf y/\Vert\mathbf y\Vert\rangle$, and $\langle\mathbf x/\Vert\mathbf
x\Vert,\mathbf z/\Vert\mathbf z\Vert\rangle$, so ordering is preserved.
\end{proof}

\begin{remark}
By transforming a query vector $\mathbf x$ into $(\mathbf x')_{\bm\beta}$ and
document vectors $\mathbf y$ into $(\mathbf y')_{\bm\beta}$, we can retrieve
documents according to the \abbr{SCM} in vector databases that only support
nearest neighbor search according to the dot product.
\end{remark}

\begin{theorem}
Let $G=(\mathbb R^n, \mathbf W_{\bm{\alpha}}, \mathbf S_{\bm{\beta}})
$ be a soft \abbr{VSM} s.t.\ $\mathbf S_{\bm{\beta}}$ is non-negative.
Let $\mathbf{x,y,y',z,z'}\in\mathbb R^n,$ and
$\mathbf{x',y'',z''}\in\mathbb R^{n+1}$ s.t. $
\mathbf{x}\not=0,\mathbf{y},\mathbf{z}>0,
(\mathbf{x}')_{\bm{\beta}'} = \left[\frac{\mathbf S\tran(\mathbf x)_{\bm\beta}}{\sqrt{\bigl(\mathbf S\tran(\mathbf x)_{\bm\beta}\bigr)\tran\mathbf S\tran(\mathbf x)_{\bm\beta}}}\:\:\:\:0\right]\tran, (\mathbf{y}')_{\bm{\beta}} = \frac{(\mathbf y)_{\bm{\beta}}}{\sqrt{\bigl((\mathbf y)_{\bm{\beta}}\bigr)\tran\mathbf S(\mathbf y)_{\bm{\beta}}}},
\newline (\mathbf{y}'')_{\bm{\beta}'} = 
\left[\bigl((\mathbf y')_{\bm{\beta}}\bigr)\tran\:\:\sqrt{1-\bigl((\mathbf y')_{\bm{\beta}}\bigr)\tran(\mathbf y')_{\bm{\beta}}}\right]\tran,
(\mathbf{z}')_{\bm{\beta}} = \frac{(\mathbf z)_{\bm{\beta}}}{\sqrt{\bigl((\mathbf z)_{\bm{\beta}}\bigr)\tran\mathbf S(\mathbf z)_{\bm{\beta}}}},$ 
and 
$(\mathbf{z}'')_{\bm{\beta}'} = \left[\bigl((\mathbf z')_{\bm{\beta}}\bigr)\tran\:\:\sqrt{1-\bigl((\mathbf z')_{\bm{\beta}}\bigr)\tran(\mathbf z')_{\bm{\beta}}}\right]\tran,
$ 
where 
$\bm\beta' = \bm\beta\cup\{[0 \ldots 0\:1]\tran\in\mathbb R^{n+1}\}.$
Then 
$\langle\mathbf x/\Vert\mathbf x\Vert, \mathbf y/\Vert\mathbf y\Vert\rangle\leq\langle\mathbf x/\Vert\mathbf x\Vert, \mathbf z/\Vert\mathbf z\Vert\rangle$
iff \newline
$\frac{\bigl((\mathbf{x}')_{\bm{\beta}'}\bigr)\tran(\mathbf{y}'')_{\bm{\beta}'}}{\sqrt{\bigl((\mathbf{x}')_{\bm{\beta}'}\bigr)\tran(\mathbf{x}')_{\bm{\beta}'}}\sqrt{\bigl((\mathbf{y}'')_{\bm{\beta}'}\bigr)\tran(\mathbf{y}'')_{\bm{\beta}'}}}\leq
\frac{\bigl((\mathbf{x}')_{\bm{\beta}'}\bigr)\tran(\mathbf{z}'')_{\bm{\beta}'}}{\sqrt{\bigl((\mathbf{x}')_{\bm{\beta}'}\bigr)\tran(\mathbf{x}')_{\bm{\beta}'}}\sqrt{\bigl((\mathbf{z}'')_{\bm{\beta}'}\bigr)\tran(\mathbf{z}'')_{\bm{\beta}'}}}.$
\end{theorem}

\begin{proof}
$\bigl((\mathbf x')_{\bm\beta'}\bigr)\tran(\mathbf x')_{\bm\beta'}=1$.
Since $\mathbf S$ is non-negative, and $(\mathbf{y})_{\bm{\beta}}>0$, $\sqrt{\bigl((\mathbf
y)_{\bm{\beta}}\bigr)\tran\mathbf S(\mathbf y)_{\bm{\beta}}}\geq\sqrt{\bigl((\mathbf
y)_{\bm{\beta}}\bigr)\tran(\mathbf y)_{\bm{\beta}}}$ and therefore $\bigl((\mathbf
y')_{\bm\beta'}\bigr)\tran(\mathbf
y')_{\bm\beta'}\leq1$, and $\bigl((\mathbf y'')_{\bm\beta'}\bigr)\tran(\mathbf
y'')_{\bm\beta'}=1$~\cite[sec.~4.2]{neyshabur2015symmetric}. Therefore:
\begin{multline*}
\frac{\bigl((\mathbf{x}')_{\bm{\beta}'}\bigr)\tran(\mathbf{y}'')_{\bm{\beta}'}}{\sqrt{\bigl((\mathbf{x}')_{\bm{\beta}'}\bigr)\tran(\mathbf{x}')_{\bm{\beta}'}}\sqrt{\bigl((\mathbf{y}'')_{\bm{\beta}'}\bigr)\tran(\mathbf{y}'')_{\bm{\beta}'}}}
= \bigl((\mathbf{x}')_{\bm{\beta}'}\bigr)\tran(\mathbf{y}'')_{\bm{\beta}'}
\\
= \frac{\bigl((\mathbf x)_{\bm\beta}\bigr)\tran\mathbf S(\mathbf y)_{\bm{\beta}}}{\sqrt{\bigl(\mathbf S\tran(\mathbf x)_{\bm\beta}\bigr)\tran\mathbf S\tran(\mathbf x)_{\bm\beta}}\sqrt{\bigl((\mathbf y)_{\bm{\beta}}\bigr)\tran\mathbf S(\mathbf y)_{\bm{\beta}}}}.
\end{multline*}
From Lemma~\ref{lemma:similarity-inner-product}, this equals $\langle\mathbf
x/\Vert\mathbf x\Vert, \mathbf y/\Vert\mathbf y\Vert\rangle$
except for the missing term
$\sqrt{\bigl(\mathbf{(x)}_{\bm\beta}\bigr)\tran\mathbf
S\mathbf{(x)}_{\bm\beta}}$, and the extra term $\sqrt{\bigl(\mathbf
S\tran(\mathbf x)_{\bm\beta}\bigr)\tran\mathbf S\tran(\mathbf
x)_{\bm\beta}}$ in the divisor. The terms are constant in both
$\langle\mathbf x/\Vert\mathbf x\Vert,\mathbf y/\Vert\mathbf y\Vert\rangle$,
and $\langle\mathbf x/\Vert\mathbf x\Vert,\mathbf z/\Vert\mathbf
z\Vert\rangle$, so ordering is preserved.
\end{proof}

\begin{remark}
By transforming a query vector $\mathbf x$ into $(\mathbf x')_{\bm\beta'}$ and
document vectors $\mathbf y$ into $(\mathbf y'')_{\bm\beta'}$, we can retrieve
documents according to the \abbr{SCM} in vector databases that only support
nearest neighbor search according to the cosine similarity.

Whereas most vector databases are designed for storing
low-dimensional and dense vector coordinates, document vectors have the dimension
$n$, which can be in the millions for real-world corpora such as the English
Wikipedia. Apart from that, a document contains only a small fraction of the terms
in the vocabulary, which makes the coordinates extremely sparse. Therefore, the
coordinates need to be converted to a dense low-dimensional representation,
using e.g.\ the latent semantic analysis (\abbr{LSA}), before they are stored
in a vector database or used for queries.
\end{remark}

Unlike vector databases, inverted-index-based search engines are built around a
data structure called the \term{inverted index},
which maps each term in our vocabulary to a list of documents (a
\term{posting}) containing the term. Documents in a posting are sorted by a
common criterion. The search engine tokenizes a text query into terms,
retrieves postings for the query terms, and then traverses the postings,
computing similarity between the query and the documents.\looseness=-1

We can directly replace the search engine's document similarity formula
with the formula for the inner product from
Lemma~\ref{lemma:similarity-inner-product}, or the formula for the \abbr{SCM}.
After this straightforward change, the system will still only retrieve
documents that have at least one term in common with the query. Therefore,
we first need to \term{expand} the query vector $\mathbf x$ by computing $((\mathbf
x)_{\bm\beta})\tran\mathbf S$ and retrieving postings for all terms
corresponding to the nonzero coordinates in the expanded vector. The expected
number of these terms is $\mathcal{O}(mC)$, where $m$ is the number of non-zero
elements in $(\mathbf x)_{\bm\alpha}$, and $C$~is the maximum number of
non-zero elements in any column of~$\mathbf S$.  Assuming $m$ and $C$ are
bounded by a constant, $\mathcal{O}(mC)=\mathcal{O}(1)$.

\section{Conclusion and future work}
\label{sec:similarity-conclusion}
In this paper, we examined the soft vector space model (\abbr{VSM})
of \newcite{sidorov2014soft}. We restated the definition, we proved a tighter
lower time complexity bound of $\mathcal{O}(n^3)$ for a related
orthonormalization problem, and we showed how the inner product, and the soft
cosine measure between document vectors can be efficiently computed in
general-purpose vector databases, in the inverted indices of text search
engines, and in other applications. To complement this paper, we also provided
an implementation of the \abbr{SCM} to \ifreview\else Gensim\footnote{%
  See \url{https://github.com/RaRe-Technologies/gensim/}, pull requests 1827,
  and 2016.%
}~\cite{rehurek_lrec}, \fi a free open-source natural language processing
library.\looseness=-1

In our remarks for Theorem~\ref{thm:similarity-inner-product}, we discuss
strategies for making no column of matrix~$\mathbf{S}$ contain more
than $C$ non-zero elements. Future research will evaluate their performance on
the semantic text similarity task with public datasets. Various choices of
the matrix $\mathbf{S}$ based on word embeddings, Levenshtein distance,
thesauri, and statistical regression as well as metric matrices from previous
work~\cite{mikawa2011proposal} will also be evaluated both amongst themselves
and against other document similarity measures such as the \abbr{LDA}, \abbr{LSA},
and \abbr{WMD}.

\ifreview\else
\subsubsection*{Acknowledgements}
We gratefully acknowledge the support by \abbr{TAČR} under the Omega program,
project \abbr{TD03000295}. We also sincerely thank three anonymous reviewers
for their insightful comments.
\fi

\printbibliography

\end{document}